\documentclass{article}
\usepackage{ijcai22}

\usepackage{times}
\usepackage{soul}
\usepackage{url}
\usepackage[hidelinks]{hyperref}
\usepackage[utf8]{inputenc}
\usepackage[small]{caption}
\usepackage{graphicx}
\usepackage{tcolorbox}
\usepackage{tikz}
\usetikzlibrary{trees}

\usepackage{amsmath,amsthm,amsfonts,amssymb,bm,bbm,mathtools,nicefrac}\usepackage{booktabs}
\usepackage{algorithm}
\usepackage{algorithmic}
\newtheorem{example}{Example}
\newtheorem{theorem}{Theorem}
\newtheorem{lemma}{Lemma}
\newtheorem{proposition}{Proposition}
\newtheorem{corollary}{Corollary}

\newtheorem{claim}{Claim}

\theoremstyle{definition}
\newtheorem{definition}{Definition}[section]

\title{Explaining Preferences by Multiple Patterns in Voters' Behavior}

\author{
Sonja Kraiczy $^1$
\and
Edith Elkind$^1$
\affiliations
$^1$ Department of Computer Science, University of Oxford\\
\emails
Sonja.Kraiczy@merton.ox.ac.uk,
eelkind@cs.ox.ac.uk
}

\begin{document}

\maketitle

\begin{abstract} 
In some preference aggregation scenarios, voters' preferences are highly structured: e.g., the set of candidates may have one-dimensional structure (so that voters' preferences are single-peaked) or be described by a binary decision tree (so that voters' preferences are group-separable). However, sometimes a single axis or a decision tree is insufficient to capture the voters' preferences; rather, there is a small number $k$ of axes or decision trees such that each vote in the profile is consistent with one of these axes (resp., trees). In this work, we study the complexity of deciding whether voters' preferences can be explained in this manner. For $k=2$, we use the technique developed by Yang~[2020] 
in the context of single-peaked preferences to obtain a polynomial-time algorithm for several domains: value-restricted preferences, group-separable preferences, and a natural subdomain of group-separable preferences, namely, caterpillar group-separable preferences. For $k\ge 3$, the problem is known to be hard for single-peaked preferences; we show that this is also the case for value-restricted and group-separable preferences. Our positive results for $k=2$ make use of forbidden minor characterizations of the respective domains; in particular, we establish that the domain of caterpillar group-separable preferences admits a forbidden minor characterization.
\end{abstract}

\section{Introduction}
A country X is about to have a general election. Each political party in X can be identified with a position on the left-to-right political spectrum. In addition, each party has also formulated a covid-19 policy, regarding issues such as vaccination requirements, school closures and mask mandates. These policies provide an alternative ordering of the parties, from those that support stringent measures to those that are opposed to any restrictions. The two orderings are quite different: e.g., while one of the left-wing parties believes that restrictions are harmful to their electorate, another one supports extreme virus control measures. 

Alice, Bob, Carol and Dave are planning to vote in this election.
Alice and Bob's preferences are driven by the parties' positions on covid-19, even though their own preferences concerning covid-19 measures are very different: Alice strongly supports the continued use of NPIs, while Bob is opposed to them. However, both Alice and Bob completely ignore the parties' positions on the traditional left-to-right spectrum. In contrast, Carol and Dave believe that the pandemic will be over soon in any case, and rank the parties 
based on their social and economic policies. Thus, in this case collective preferences are driven by two axes on the set of candidates, 
with each voter's ranking being consistent with one of the axes (and ignoring the other axis).

Now, suppose that we expect the collective preferences to have this general shape (potentially with $k\ge 2$ axes), but we do not know the underlying $k$ orderings: can we identify them, and the associated partition of the voter set, in polynomial time? This question is just as relevant for other notion of structure: while the idea of preferences being consistent with an axis is captured by the mathematical concept of single-peaked preferences \cite{B1948}, we can also consider single-crossing \cite{M71,R77}
or group-separable preferences \cite{I64}. The latter domain, which has received comparatively less attention in the computational social choice literature (see, however, the recent work of Faliszewski et al.~[\citeyear{FKO21}]), 
consists of preference profiles 
that can be explained by binary decision trees:
each alternative is characterized by a set of binary attributes, 
each voter has a preferred value for each attribute, 
and there is a binary tree whose vertices are labeled
with attributes that guides the voters' decision-making process
(we present formal definitions in Section~\ref{sec:prelim}).

The problem of partitioning the input profile into $k$ 
single-peaked profiles has been considered
by \citeauthor{ELP13}~[\citeyear{ELP13}], who obtained NP-hardness results for every $k\ge 3$, but left the case $k=2$ open. This open
question was highlighted by \citeauthor{JPE18}~[\citeyear{JPE18}] (who also
study the analogue of this problem for single-crossing preferences, and obtain a polynomial-time algorithm for $k=2$) and subsequently resolved by 
\citeauthor{YY20}~[\citeyear{YY20}], who showed the for $k=2$
this problem admits a polynomial-time algorithm.
However, its variant for group-separable
preferences (where we seek $k$ binary decision trees that `explain' the input
profile) has not been considered before.

\paragraph{Our Contribution}
The proof by \citeauthor{YY20}~[\citeyear{YY20}] is based on the 
characterization of the single-peaked preferences in terms of forbidden
minors: as shown by \citeauthor{BH11}~[\citeyear{BH11}], there is a small set of constant-size preference profiles
such that a profile is single-peaked if and only if it does not contain
a subprofile that is isomorphic to one of the profiles in this set.
We observe that this approach extends to 
several other domains that admit a forbidden minor characterization, 
including, in particular, group-separable preferences. 
Further, we consider a natural subdomain of group-separable preferences, namely, the 
caterpillar group-separable domain: it consists of profiles for which the underlying binary tree is caterpillar-shaped, i.e., each binary decision pitches a single candidate against all other candidates.
We provide a characterization of this domain in terms of forbidden minors, thereby showing that our algorithm applies to this domain as well. To complement these results, we show that
the partitioning problem is NP-hard for group-separable preferences
(and several other related domains) for every value of $k\ge 3$.

There are two reasons why we think that our results are interesting.
First, a binary decision tree for a group-separable profile helps us 
understand the structure of the alternative space; this is still the case
where the profile is `explained' by two trees. Second, 
from a more practical perspective, there are voting 
problems that are computationally hard for general preferences, 
but admit polynomial-time algorithms for structured preferences;
a prominent example is the algorithm for the Chamberlin--Courant rule
for single-peaked preferences \cite{BSU13}, which relies on knowing the axis. While we do not yet know if similar results can be obtained
for profiles that can be partitioned into a small number of structured
profiles (see the work of \citeauthor{MSV17}~[\citeyear{MSV17}] for some contributions in this spirit), our results provide a promising starting point and a necessary ingredient for such algorithms.

\paragraph{Related Work}
Our work belongs to a stream of research on the complexity of identifying 
nearly structured profiles (i.e., profiles that can be made single-peaked/single-crossing/group-separable/etc.~by small modifications),
which was initiated by \citeauthor{BCW16}~[\citeyear{BCW16}] and \citeauthor{ELP13}~[\citeyear{ELP13}];
see also the work of \citeauthor{JPE18}~[\citeyear{JPE18}] and
\citeauthor{LPE19}~[\citeyear{LPE19}], and the survey by
\citeauthor{ELP-survey}~[\citeyear{ELP-survey}]. 

Several structured domains can be characterized by a small set
of forbidden minors: this is the case for single-peaked and group-separable preferences \cite{BH11} and for single-crossing preferences \cite{BCW13}. In addition some domains are directly defined in terms of forbidden minors (i.e., the best/medium/worst/value-restricted domains \cite{Sen}). Our minor-based approach is conceptually similar to the work of \citeauthor{EL14}~[\citeyear{EL14}], who provide approximation algorithms for voter/candidate deletion towards structured preferences, for all domains that can be characterized by constant-size forbidden minors.

\citeauthor{karpov}~[\citeyear{karpov}]
and \citeauthor{FKO21}~[\citeyear{FKO21}] initiated the algorithmic analysis of group-separable preferences. In particular, to the best of our knowledge, \citeauthor{FKO21}~[\citeyear{FKO21}] are the first to discuss the domain of caterpillar group-separable preferences; however, their primary focus is on the complexity of voting problems for such preferences rather than on the structure of this domain per se. 

\section{Preliminaries}\label{sec:prelim}
Let $C$ be a finite set of {\em candidates}.
A {\em vote} over $C$ is a linear order over $C$. Given a vote $v$ over $C$
and two candidates $a, b\in C$, we write $a\succ_v b$ to denote that $a$
is ranked above $b$ in $v$. We extend this notation to sets: $A\succ_v B$
means that $v$ ranks all elements of $A$ above all elements of $B$.
A {\em preference profile} $P$ over a candidate set $C$ is a list of 
votes over $C$.

Given a vote $v$ over $C$ and a subset of candidates $A\subseteq C$, 
the {\em restriction of $v$ to $A$} is the vote $v|_A$ over $A$
such that for all $a, b\in A$ it holds that $v|_A$ ranks $a$ above $b$
if and only if $v$ ranks $a$ above $b$. 
Given a profile $P=(v_1, \dots, v_n)$ over $C$ and a subset of candidates $A\subseteq C$, the 
{\em restriction of $P$ to $A$} is
the profile $P|_A=(u_1, \dots, u_n)$ such that $u_i=v_i|_A$ for each $i\in [n]$.
A profile $P'$ over $A$ is a {\em subprofile} of a profile $P$ over $C$ if $A\subseteq C$ and $P'$ is obtained by removing zero or more votes from $P|_A$.


We will consider several special classes of preferences.

\begin{definition}[Single-peaked preferences]
Let $\lhd$ be a linear order over a candidate set $C$.
A vote $v$ over $C$ is {\em single-peaked on $\lhd$} 
if for every triple $a, b, c\in C$ such that $a\lhd b\lhd c$
it holds that $b\succ_v a$ or $b\succ_v c$.
A profile $P$ over $C$ is {\em single-peaked on $\lhd$} if every vote in $P$ is single-peaked on $\lhd$. A profile $P$ is {\em single-peaked} if there exists an  order $\lhd$ over $C$ such that $P$ is single-peaked on $\lhd$;
this order is referred to as the {\em axis}.
\end{definition}
Single-peaked preferences model settings where all candidates can be ordered on a left-to-right axis, and each voter has a favorite point on this axis and ranks candidates on either side of their favorite point in order of increasing 
distance from this point: e.g., if the vote is over tax rates, a voter whose most preferred tax rate is 22\% prefers 20\% to 15\% and 27\% to 40\% (but may prefer 27\% over 20\%). 
\begin{definition}[Group-separable preferences]
A profile $P$ over a candidate set $C$ is {\em group-separable} if every set of candidates $A \subseteq C$ has a proper subset $B \subset A$ such that for every vote $v \in P$ we have either $B \succ_v (A\setminus B)$ or $(A\setminus B) \succ_v B$.
\end{definition}
Equivalently, group-separable profiles can be defined in terms of binary decision trees as follows.
An {\em ordered binary tree} is a rooted tree such that each internal node has two children, 
one of these children is designated as the left child, 
and the other is designated as the right child.
Given an ordered binary tree $T$ whose leaves are labeled with elements of $C$,
we say that a vote $v$ over $C$ is {\em $T$-consistent} if for each internal node $x$
of $T$ it holds that either $v$ ranks all candidates
in the left subtree of $x$ over all candidates in the right subtree of $x$,
or vice versa. We say that a profile $P$ is {\em $T$-consistent} if every vote in $P$ is $T$-consistent. The following proposition is implicit in prior work (see, e.g. \cite{karpov}); for completeness, 
we provide a proof in Appendix~\ref{app:prelim}.

\begin{proposition}\label{prop:gs-tree}
 A profile $P$ is group-separable if and only if it is $T$-consistent for some ordered binary tree $T$.
\end{proposition}
 
One can then think of internal nodes of $T$ as binary attributes:
all candidates in the left subtree of $x$ possess the attribute associated with $x$,
while all candidates in the right subtree of $x$ do not possess it. Each voter views each attribute as desirable or undesirable and forms their ranking accordingly, 
by starting at the root of the tree and moving downwards. Thus, the tree $T$ in the definition of group-separable preferences plays a similar role to the axis $\lhd$ in the definition of single-peaked preferences: they both `explain' the rationale behind the voters' decision-making.

A {\em restricted domain} is a collection of preference profiles; e.g., we will speak of the domain of all single-peaked profiles (denoted by {\sc SP}) and the domain of all group-separable profiles (denoted by {\sc GS}).
A restricted domain $X$ is {\em hereditary} if for every profile $P\in X$ it holds that every subprofile of $P$ is in $X$. It is immediate 
from the definitions 
that both SP and GS are hereditary.

Two profiles are said to be {\em isomorphic} if they can be obtained from each other by renaming candidates and/or reordering votes. A {\em $p\times q$ minor} is a profile that contains $p$ votes and $q$ candidates. We say that a profile $P$ {\em contains a minor $Q$} if there is a subprofile of $P$ that is isomorphic to $Q$. 
A restricted domain $X$ can be {\em characterized by a set of forbidden minors} if there is a set of minors $\mathcal Q$ such that a profile $P$ belongs to $X$ if and only if it does not contain a minor in $\mathcal Q$; we refer to the set $\mathcal Q$ as {\em the set of forbidden minors for $X$}.
\begin{proposition}\label{LL}\cite{LL17}
A restricted domain can be characterized by a (possibly infinite) set of forbidden minors if and only if it is hereditary.
\end{proposition}
There are well-studied restricted domains that are explicitly defined in terms of forbidden minors.
\begin{definition}
For $j=1, 2, 3$, we say that a 3-by-3 minor $Q$ is a {\em $j$-minor} if in $Q$ each 
candidate appears in the $j$-th position exactly once (note that for each $j=1, 2, 3$ there are several $j$-minors that are not pairwise isomorphic).
A profile $P$ is \textit{best-/medium-/worst-restricted} if it does not contain any $1$-(respectively, $2$-, $3$-)minors. 
We say that a profile is {\em value-restricted} if it is simultaneously
best-, medium-, and worst-restricted. We will denote the restricted domains
that consist of best-/medium-/worst-/value-restricted profiles by, respectively, 
{\sc BR}, {\sc MR}, {\sc WR}, and {\sc VR}.
\end{definition}
By definition, each of the domains BR, MR, WR and VR can be characterized by a finite set of forbidden minors. 
It has been shown that the domains SP and GS admit such characterizations, too;
we give them below, as they are vital both for our hardness and 
for our easiness results.
\begin{theorem}\cite{BH11}\label{characsp}
A profile $P$ is in {\sc SP} if and only if 
it is worst-restricted 
and does not contain any of the four $2\times 4$ minors given by
$$
\{a,d\} \succ_u b \succ_u c, \qquad \{c,d\} \succ_v b \succ_v a.
$$
\end{theorem}

\begin{theorem}\cite{BH11}\label{thm:gs-minor}
A profile $P$ is in {\sc GS} if and only if 
it is medium-restricted and does not contain the $2\times 4$ minor given by
$$
a \succ_u b \succ_u c \succ_u d, \qquad b \succ_v d \succ_v a \succ_v c.
$$
\end{theorem}

For each restricted domain $X$ and a positive integer $k$, we define the following decision problem.
\begin{tcolorbox}
\textsc{$X$ Voter $k$-Partition}:\\
\textbf{Input}: A profile $P$.\\
\textbf{Question}: Can $P$ be partitioned into $k$ profiles $P_1,\ldots, P_k$
so that $P_i$ is in $X$ for each $i\in [k]$? 
\end{tcolorbox}

\section{Partitioning Voters into Two Groups}\label{sec:k=2}

The main result of this section is the following theorem.
\begin{theorem}\label{thm:2axesinP}
Let $X$ be a restricted domain such that for some constant $\ell \in \mathbb N$,
$X$ can be characterized by a finite set $\mathcal Q$ of forbidden minors, 
where each minor in $\mathcal Q$ is either a $2\times \ell$ minor 
or a $j$-minor for some $j=1, 2, 3$.
Then \textsc{$X$ Voter $2$-Partition} admits a polynomial-time algorithm.
\end{theorem}
Note that the condition in the statement of Theorem~\ref{thm:2axesinP}
is satisfied by all restricted domains defined in Section~\ref{sec:prelim}.
Hence, we obtain the following corollary.
\begin{corollary}\label{cor:2axesinP}
The problem \textsc{$X$ Voter $2$-Partition} is in {\textrm P}
for $X\in\{$GS, BR, MR, WR, VR$\}$.
\end{corollary}
For the SP domain, an analogue of Corollary~\ref{cor:2axesinP} has been established by \citeauthor{YY20}~[\citeyear{YY20}]; our techniques are very similar to his.
\footnote{Our results were obtained independently from \citeauthor{YY20}~[\citeyear{YY20}], as we were not aware of his work.}
\begin{proof}[Proof sketch]
Consider a domain $X$ that satisfies the conditions in the theorem statement, and let $\mathcal Q$ be the respective set of forbidden minors. Fix a profile $P$.
Note that if $P$ contains multiple copies of a vote $v$,
we can remove all but one copy without changing the answer.
Therefore, we can assume that $P$ does not contain two identical
votes. Also, the order of votes in $P$ does not matter.
Thus, from now on we will treat $P$ as a set of votes.
Our goal, then, is to decide whether we can partition $P$ as 
$P=U \cup V$ so that $U, V\in X$.

We first explain how $i$-minors, $i=1, 2, 3$, induce a partition of $P$ into three subsets. Fix a candidate triple $T=\{a, b, c\}\subseteq C$ and $i\in\{1, 2, 3\}$. 
For each $x\in T$, let $P^i_{T,x}$ be the set of votes $v\in P$ such that $x$ appears in the $i$-th position in $v|_T$.
We will say that $T$ is {\em $i$-dangerous} for $P$ if 
$\mathcal Q$ contains an $i$-minor and
$P^i_{T,a},P^i_{T,b}, P^i_{T,c}\neq \varnothing$.
Our analysis relies on the following lemma.

\begin{lemma}\label{cl:onepertriple}
Suppose $\mathcal Q$ contains an $i$-minor, $i\in [3]$. Suppose 
we can partition $P$ as $P=U\cup V$ so that $U, V\in X$. 
Then for every $i$-dangerous triple $T=\{a,b,c\}\subset C$
there is at most one candidate $a\in T$ such that both 
$U\cap P^i_{T,a}\neq \varnothing$ and $V\cap P^i_{T,a}\neq \varnothing$ hold.
\end{lemma}
\begin{proof}
Suppose that for some $i$-dangerous triple $T=\{a, b, c\}$ there are two candidates 
$a,b\in T$ such that the sets $P^i_a$ and $P^i_b$ both have non-empty intersection with each of $U$ and $V$. There is a set $W\in\{U, V\}$ 
with $P^i_{T,c}\cap W\neq\varnothing$. But then $W$ contains a forbidden $i$-minor, so it is not in~$X$. 
\end{proof}



We start by checking, for each of the $\binom{m}{3}$ triples 
$T=\{a, b, c\}\subset C$, 
whether for each $i$-minor in $\mathcal Q$ 
each of the sets $P^i_{T,a}$, $P^i_{T,b}$ and $P^i_{T,c}$ 
is in $X$. Note that, as $X$ can be characterized by a finite set of forbidden minors, and the size of each minor in $\mathcal Q$ is bounded by a constant, 
we can check in polynomial time whether a given profile is in $X$.
If for some $i$-dangerous triple $T$, and forbidden $i$-minor in $\mathcal Q$ 
at least two 
of the sets $P^i_{T,x}$, $x\in T$, (say, $P^i_{T,a}$ and $P^i_{T,b}$) 
are not in $X$, then we report that $P$ 
is a no-instance; this decision is correct by Lemma~\ref{cl:onepertriple}. 
Therefore, in what follows, for each forbidden $i$-minor $Q\in \mathcal Q$
over a triple $T=\{a, b, c\}$, we assume that
at most one of the sets $P^i_{T,a}$, $P^i_{T,b}$ and $P^i_{T,c}$ is not 
in $X$.
We now split our analysis into two cases.\\

{\noindent\bf Case 1:} For every $i$-dangerous triple $T\subset A$ and every forbidden $i$-minor of $X$, exactly one  of $P^i_{T,x}$, $x\in T$, is not in $X$.
In this case, we construct a graph $G$ with vertex set $P$ 
so that $G$ is bipartite if and only if $P$ is a yes-instance of our problem.
We define the edge set $E$ of the graph $G$ as follows.
Given two votes $u, v$ such that $(u, v)$ induces a
$2\times \ell$ forbidden minor, we add the edge $\{u, v\}$ to $E$; we refer to edges
of this type as {\em $2\times \ell$-edges}.
For every forbidden $i$-minor of $X$ and  each $i$-dangerous triple $T=\{a, b, c\}$, if $P^i_{T,c}$ is not in $X$,
then for every pair of votes $(u, v)$ with $u\in P^i_{T,a}$, $v \in P^i_{T,b}$
we add the edge $\{u, v\}$ to $E$; we refer to edges of this type as {\em $i$-edges}.
It is not hard to see that $G$ is bipartite if and only if $P$
is a yes-instance of our problem; refer to Appendix \ref{appendixa} for details.\\

{\noindent\bf Case 2:} There exists an $i$-dangerous triple $T=\{a, b, c\}\subset C$ such that for each $x\in T$ the set $P^i_{T,x}$ is in $X$.

In this case we construct three instances of 2-SAT and show that at least one of them is satisfiable if and only if $P$ is a yes-instance of our problem. 
Recall that an instance
of 2-SAT is given by a set of Boolean variables, which take values in $\{T, F\}$ 
and a collection of clauses of the form $x\lor y$, where $x$ and $y$ are (not necessarily distinct) literals
(i.e., variables or negations of variables). 
It is satisfiable if we can assign
values to all variables so that at least one literal in each clause is satisfied
(i.e., takes value $T$).
We can decide in polynomial time if a given instance of 2-SAT is satisfiable
\cite{sipser}.

For each pair of candidates $\{a,b\}\subset T$ we construct a 2-SAT instance $I_{a, b}$ that is satisfiable if and only if $P$ can be partitioned as $V=U\cup V$ so that $U$ and $V$ are in $X$, with
$P^i_{T,a}\subset U$ and $P^i_{T,b}\subset V$. To this end, we use Lemma~\ref{cl:onepertriple}.

For each $v\in P^i_{T,c}$, we create a Boolean variable $x_v$;
we interpret $x_v=T$ as $v\in U$ and $x_v=F$ as $v\in V$. 
We create the following clauses:
\begin{itemize}
    \item  For each $v\in P^i_{T,c}$, if there exists a $u \in P^i_{T,a}$ such that $u, v$ induce a $2\times \ell$ forbidden minor, we add the clause $\neg x_v$. 
    Similarly, if there exists a $u \in P^i_{T,b}$ such that $u, v$ induce a 
    $2\times \ell$ forbidden minor, we add the clause~$x_v$.
    \item For each $v\in P^i_{T,c}$, if there are $u,w \in P^i_{T,a}$ such that $u, v, w$ induce a $j$-minor in $\mathcal Q$ for $j\in[3]$, we add the clause $\neg x_v$; if there are $u,w \in P^i_{T,b}$ such that $u, v, w$ induce a forbidden $j$-minor, we add the clause~$x_v$.
    \item For each pair of votes  $u,v\in P^i_{T,c}$ if there is a vote $w \in P^i_{T,a}$ such that $u, v, w$ induce a $j$-minor in $\mathcal Q$, 
    we add the clause $(\neg x_u \lor \neg x_v)$, and 
    if there is a vote $w \in P^i_{T,b}$ such that $u, v, w$ induce 
    a $j$-minor in $\mathcal Q$, we add the clause $(x_u \lor x_v)$.
\end{itemize}

Suppose that $P$ can be partitioned as $U\cup V$ so that $U, V\in X$.
We know that $T$ is $i$-dangerous 
and there exist $a, b\in T$ such that 
$P^i_{T,a}\subseteq U$, $P^i_{T,b}\subseteq V$ by Lemma~\ref{cl:onepertriple}. 
We claim that 
$I_{a, b}$ is satisfiable. Indeed, for each $v\in P^i_{T,c}$, 
let $x_v=T$ if $v\in U$ and let $x_v=F$ if $v\in V$. Consider a clause
of the form $x_v$. For this clause not to be satisfied, it has to be the case
that $v\in V$. But then $I_{a, b}$ can only contain this clause
if there exists a $u \in P^i_{T,b}\subseteq V$ such that $u, v$ induce
a $2\times \ell$ forbidden minor or
if there exist $u,w \in P^i_{T,b}\subseteq V$ such that 
$u, v, w$ induce a $j$-minor in $\mathcal Q$; 
in either case, we obtain a contradiction with $V\in X$. 
Similarly, for a clause of the form $\neg x_v$ not to be satisfied, 
it has to be the case that $v\in U$ and $U$ contains a forbidden minor.
Further, for a clause $(x_u \lor x_v)$ not to be satisfied, 
it has to be the case that $u, v\in V$ and $V$ contains a forbidden minor,
and for a clause $(\neg x_u \lor \neg x_v)$ not to be satisfied, 
it has to be the case that $u, v\in U$ and $U$ contains a forbidden minor.
Thus, the truth assignment described above satisfies $I_{a, b}$.

Conversely, suppose that there is a pair $\{a, b\}\subseteq T$
such that $I_{a, b}$ is satisfied, and let $(x^*_v)_{v\in P^i_{T,c}}$
be a satisfying assignment for it.
We then construct $U$, $V$
by setting $U=P^i_{T,a} \cup \{v\in P^i_{T,c}: x^*_v=T\}$, 
$V=P^i_{T,b} \cup \{v\in P^i_{T,c}: x^*_v=F\}$. 
We claim that $U$ and $V$ are in $X$. 
Indeed, consider $U$, and suppose that
it contains a $2\times \ell$ forbidden minor involving votes $u$
and $v$. Since $P^i_{T,a}$
and $P^i_{T,b}$ are in $X$, 
we can assume without loss of generality that $u\in P^i_{T,a}$, $v\in P^i_{T,b}$.
But in that case $I_{a, b}$ contains the clause $\neg x_v$, 
so we must have $x^*_v=F$, a contradiction with $v$ being placed
in $U$. Now, suppose that $U$ contains a $j$-minor in $\mathcal Q$, $j\in [3]$,
involving votes $u$, $v$ and $w$. Since $P^i_{T,a}$
and $P^i_{T,b}$ are in $X$, 
we can assume without loss of generality that 
either (1) $u, w\in P^i_{T,a}$, $v\in P^i_{T,c}$ or
(2) $w\in P^i_{T,a}$, $u, v\in P^i_{T,c}$.
But then in case (1) the instance $I_{a, b}$
contains the clause $\neg x_v$, 
so we must have $x^*_v=F$, and 
in case (2) the instance $I_{a, b}$
contains the clause $\neg x_u \lor \neg x_v$, 
so we must have $x^*_u=F$ or $x^*_v=F$.
In either case, we get a contradiction with how $U$
is constructed. 
We conclude that $U$ does not contain forbidden minors 
and therefore it is in $X$;
by the same argument, $S_2$ is in $X$ as well. 
To summarize, our algorithm needs to consider $O(m^3)$ profiles, 
check whether each of them is in $X$, and then
either construct a graph and decide whether it is bipartite or solve
three instance of 2-SAT. Thus, our algorithm runs in polynomial time.
\end{proof}
\section{Partitioning Voters into at Least Three Groups}\label{sec:k>2}
\citeauthor{ELP13}~[\citeyear{ELP13}] show that \textsc{SP Voter $k$-Partition} is NP-complete even when $k\geq 3$. Their reduction  is from \textsc{$k$-Partition Into Cliques}. 
This problem, which is known to be NP-complete \cite{karp72}, 
is defined as follows.
\begin{tcolorbox}
\textsc{$k$-Partition Into Cliques}:\\
\textbf{Input}: A graph $G=(V_G,E_G)$.\\
\textbf{Question}: Can we partition $V_G$ into $k$ sets such that each set of vertices induces a clique on $(V_G,E_G)$.
\end{tcolorbox}

In the following we show that \textsc{VR Voter $k$-Partition} 
and \textsc{GS Voter $k$-Partition} are NP-complete, too. 
While our proof also proceeds
by a reduction from \textsc{$k$-Partition Into Cliques},
our argument is quite different: we use $3\times 3$ minors, 
whereas \citeauthor{ELP13}~[\citeyear{ELP13}] use $2\times 4$ 
minors. The advantage of our approach is that it also applies to the VR domain, whose forbidden minor characterization does not use $2\times 4$ minors.

We start by considering the domains BR, MR, WR and VR. Then we explain why our proof approach also works for the GS domain. It is immediate that 
{\sc $X$ Voter $k$-Partition} is in NP for each domain $X$
that we consider, so in what follows we focus on NP-hardness proofs.
 
 \begin{theorem}\label{thm:vr-npcomplete}
 \textsc{X Voter $k$-Partition} for $X\in\{$BR, MR, WR, VR$\}$ is 
 {\em NP}-complete for each $k\geq 3$.
 \end{theorem}
\begin{proof}
Given a graph $G=(V,E)$, we first create a graph $G'$ so that 
$G'$ contains $k$ cliques of size $k+2$ and
is a yes-instance of \textsc{$k$-Clique Partition} 
if and only if $G$ is.
For each $i\in [k]$,  
let $H_i$ be a clique with vertex set $U_i$, $|U_i|=k+2$.
To construct the graph $G'$, we connect
the vertices of all these cliques to all vertices of $G$. 
That is, $G'$ is the graph with vertex set 
$V' = V \cup U_1\cup\ldots\cup U_k$ and edge set
$$
E' =E \cup \{\{u,v\} : u\in V, v\in \cup_{i\in [k]} U_i\}.
$$ 
If $V'_1,\ldots,V'_k$ is a partition of $G'$ into $k$ cliques, then each $V'_i$, $i\in [k]$, restricted to $G$ is either a clique or 
an empty set.
Conversely, if $V_1,\ldots, V_k$ is a partition of $G$ into $k$ cliques, then $V'_1, \dots, V'_k$, where
$V'_i = V_i\cup U_i$, 
is a partition of $G'$ into $k$ cliques.  

We are now ready to create an instance of 
\textsc{VR voter $k$-partition}. For convenience, 
renumber the vertices of $G'$ as $u_1, \dots, u_n$.

\noindent \textbf{Instance} 
In our instance, there are three candidates for each pair 
of vertices that does
not form an edge of $G'$, i.e., for each 
$\{u_i, u_j\}\in (V'\times V')\setminus E'$
we set
$T^{i, j}=\{a^{i,j},b^{i,j},c^{i,j}\}$ and let 
$$
C=\bigcup_{\{u_i, u_j\}\in (V'\times V')\setminus E'} T^{i, j}.
$$
We set $P=(v_1, \dots, v_n)$, where $n=|V'|$.
In each vote, the triples of candidates are 
ordered according to their indices:
If $i<\ell$ or $i=\ell$ and $j<r$ then in each vote
all candidates in $T^{i, j}$ appear above
all candidates in $T^{\ell, r}$.
Further, if $\ell\neq i, j$ then in $v_\ell$ 
candidates in $T^{i, j}$
are ranked as $c^{i, j} \succ a^{i, j} \succ b^{i, j}$.
Finally, in vote $v_i$ these candidates are ranked as
$a^{i,j}\succ b^{i,j}\succ c^{i,j}$ and an $v_j$ they are ranked
as
$b^{i,j}\succ_j c^{i,j}\succ_j a^{i,j}$.

Suppose $V_1,\ldots,V_k$ is a partition of $G'$ into cliques.
We claim that for each $\ell\in [k]$ the profile $(v_i)_{u_i\in V_\ell}$ is in VR (and hence also in BR, MR and WR), i.e., it contains no $j$-minors for $j=1, 2, 3$.
Indeed, for a triple of votes $u, v, w$ and a triple of candidates 
$a, b, c$ to form a $j$-minor for some $j=1, 2, 3$, it has to be the case that $\{a, b, c\}=T^{r,s}$ for some 
$\{u_r, u_s\}\in (V'\times V')\setminus E'$ and 
$v_r, v_s\in \{u, v, w\}$, 
a contradiction with $V_\ell$ forming a clique.

Conversely, let $P_1,\ldots,P_k$ be a partition of $P$ into $k$  value-restricted profiles (the same argument works if each of these profiles is in BR, or if each of them is in MR, or in WR). 
Note that for each $\ell\in [k]$
and each $j=1, 2, 3$ the profile $P_\ell$ does not contain a $j$-minor. We will argue that each vertex set $V_\ell = \{u_i: v_i\in P_\ell\}$ forms a clique in $G'$. 

Observe first that if we have $u_r, u_s\in V_\ell$ for 
some $\ell\in [k]$ and $\{u_r, u_s\}\not\in E'$, 
then $V_\ell =\{u_r, u_s\}$.
Indeed, if $V_\ell$ contains another vertex $u_t$, where
$t\neq r, s$, then $v_t$ ranks the alternatives
in $T^{r,s}$ as 
$c^{r, s} \succ a^{r, s} \succ b^{r, s}$
and therefore $v_r, v_s, v_t$
and $T^{r,s}$ form a $j$-minor
for each $j=1, 2, 3$.

It follows that each set $V_\ell$ is either a clique in $G'$
or a pair of vertices with no edge between them. We will now use
a counting argument to rule out the latter possibility.

Recall that each of the disjoint cliques $H_1, \dots, H_k$ is of size $k+2$.
Therefore, by the pigeonhole principle, for each $j\in [k]$
there exists a set $V_{\ell(j)}$ 
such that $|V_{\ell(j)}\cap H_j|\ge 2$.
Moreover, if $j\neq j'$ then $\ell(j)\neq \ell(j')$.
Indeed, suppose that $\ell(j)=\ell(j')$ for some $j\neq j'$, 
and consider the set $V_{\ell(j)}$. It contains at least 
four distinct vertices, but it is not a clique, as there
are no edges between $H_j$ and $H_{j'}$, and we have
argued that this is not possible.

Hence, the mapping $j\mapsto \ell(j)$ is a bijection.
That is, each set $V_\ell$, $\ell\in [k]$, 
contains two vertices from the same clique.
Hence, no such set consists of two vertices
that are not connected, and we have argued 
that in this case $V_\ell$ must be a clique.
This proves our claim.
\end{proof}

We will now explain how to extend the proof of Theorem~\ref{thm:vr-npcomplete} to group-separable preferences.
\begin{theorem}\label{gsnpcomplete}
\textsc{GS Voter $k$-partition} is {\em NP}-complete for each $k\geq 3$.
\end{theorem}
\begin{proof}
We use the same reduction as in the proof of Theorem~\ref{thm:vr-npcomplete}. 
Suppose the resulting profile $P$ can be partitioned into 
$k$ profiles $P_1, \dots, P_k$ so that each $P_t$ is group-separable.
Then, in particular, each $P_t$ is in MR and hence corresponds to a clique in $G'$.

Conversely, suppose the graph $G'$ can be partitioned into $k$ cliques $V_1, \dots, V_k$, and let $P_1, \dots, P_k$ be the respective partition of $P$. 
The proof of Theorem~\ref{thm:vr-npcomplete} shows that each $P_t$ does not contain a $2$-minor. Hence, by Theorem~\ref{thm:gs-minor}
it remains to argue that each $P_t$ does not contain the $2\times 4$ minor
$a\succ_u b\succ_u c \succ_u d$, $b\succ_v d\succ_v a\succ_v c$.
Suppose for the sake of contradiction that for some $t\in [k]$ the 
profile $P_t$ contains this minor; abusing notation somewhat, 
assume that $a, b, c, d \in C$ and $u, v\in P_t$. 
Consider the triple $T^{i, j}$ 
such that $a\in T^{i, j}$. It cannot be the case that 
$d\in T^{i, j}$, because $a\succ_u b\succ_u c \succ_u d$ 
would imply $b, c\in T^{i, j}$, but we have $|T^{i, j}|=3$. 
Then $d\in T^{\ell, r}$, where $\ell>i$
or $\ell=i$ and $r>j$. But then all other voters in $P$, 
including $v$, rank $a$ above $d$, a contradiction.
\end{proof}


\section{Group-separability on a Caterpillar}\label{sec:gs-cat}
 In this section, we consider profiles that are group-separable on caterpillar graphs.
 A \textit{caterpillar} is a binary tree in which each internal node has at least one child that is a leaf. 
 Let $E$ be a caterpillar with $m$ leaves; 
 observe that is has $2m-1$ vertices. For each $i=1, \dots, m-2$, the tree $E$ has exactly one leaf at depth $i$; 
 we will denote this leaf by $c_i$, and denote the two leaves at depth $m-1$ 
 by $c_{m-1}$ and $c_m$. We will refer to $E$ by $(c_1,\ldots,c_m)$.
 In what follows, given a caterpillar of this form, 
 it will be convenient to denote the set of candidates 
 $\{c_i, \dots, c_j\}$, where $1\le i\le j\le m$, by $C_{[i, j]}$.
 
 Using this notation, we can say that $V$ is group-separable on a caterpillar $(c_1, \dots, c_m)$ if for every $v\in V$ and every $i\in [m-1]$ it holds that $c_i\succ_v C_{[i+1, m]}$ or $C_{[i+1, m]}\succ_v c_i$. 
 Let {\sc Cat-GS} denote the domain
 of all profiles that are group-separable on a caterpillar.
For proofs of the following two propositions, see Appendix~\ref{app:gs-cat}.
 \begin{proposition}\label{prop:gs-cat-charac}
A profile $P$ is group-separable on caterpillar $(c_1,c_2,\ldots,c_m)$ if and only if for every $v\in P$ 
there is a subset $C'\subseteq C$ such that 
$C'\succ_v C\setminus C'$, and $v$ ranks the candidates in $C'$ in increasing order of indices and candidates in $C\setminus C'$
in decreasing order of indices.
\end{proposition}
\begin{proposition}\label{prop:gs-cat-closed}
{\sc Cat-GS} is closed under candidate deletion.
\end{proposition}
Recall that closure under candidate deletion is necessary for a domain to admit a characterization by forbidden minors (Proposition~\ref{LL}).

\noindent\textbf{Recognition Algorithm}
The {\sc Cat-GS} domain admits a simple recognition algorithm. Let us say that a candidate is \textit{polarizing} for vote $v$ if she is ranked either first or last in $v$; let $\pi(v)$ denote the set of polarizing candidates for vote $v$. Given a profile $P$, the algorithm proceeds in $m-2$ steps. At each step, it looks for a candidate that is polarizing for all votes. If some such candidate is found, it is removed from all votes, and the algorithm proceeds to the next step. If no such candidate is found, the algorithm reports that $P$ does not belong to {\sc Cat-GS}. Now, suppose the algorithm succeeds. 
Relabel the candidates so that the candidate identified at the $j$-th step is labeled as $c_j$, and the two candidates that remain after the algorithm terminates are labeled as $c_{m-1}$ and $c_m$. Then the 
profile $P$ is caterpillar group-separable on $(c_1, \dots, c_m)$.
The correctness of the algorithm is immediate from Proposition~\ref{prop:gs-cat-closed}.



We are now ready to present our minor-based characterization of {\sc Cat-GS}.

\begin{theorem}\label{thm:gs-cat-minor}
A profile $P$ is group-separable on a caterpillar if and only if 
(1) it is medium restricted and 
(2) it does not contain any of the four $2\times 4$ forbidden minors given by
$$
 a\succ_u\{b,c\}\succ_u d \text{ and } b\succ_v\{a,d\}\succ_v c.
 $$
\end{theorem}
\begin{proof}
It is easy to see that if a profile is group-separable on a caterpillar, it
satisfies conditions~(1) and~(2);
for the formal proof, see Appendix~\ref{app:gs-cat}.

For the converse direction, we show that if our recognition algorithm fails on $P$, then $P$ contains a forbidden $2\times 4$ minor given by condition~(2) or a $2$-minor. Suppose our recognition algorithm fails at step $j$, 
$j\le m-2$, i.e., there is no candidate at that step that is polarizing for all votes. From now on, we consider the restriction of $P$ to the remaining candidates.

Consider a vote $u$, and let $a$ and $b$ be the two polarizing candidates for that vote, so that $\pi(u)=\{a, b\}$. We know that there is some vote $v$ such that $a\not\in\pi(v)$. If $b\not\in\pi(v)$
either, i.e., $\pi(v)=\{c, d\}$ and $\{a, b\}\cap\{c, d\}=\varnothing$, then the votes $u, v$
and candidates $a, b, c, d$ form a forbidden $2\times 4$ minor that satisfies condition~(2), and we are done. 

So it remains to consider the case where $\pi(v)=\{b, c\}$ for some
$c\not\in\{a, b\}$. In this case, $b$ is polarizing for both $u$ and $v$; hence, there must exist a vote $w$ such that $b\not\in\pi(w)$.
Now, if $\pi(w)=\{a, c\}$, the votes $u, v, w$ and the candidates $a, b, c$ form a $2$-minor witnessing that the profile is not medium-restricted, and hence does not belong to {\sc Cat-GS}. Thus, there exists a candidate $d\not\in\{a, b, c\}$ such that $d\in\pi(w)$.
Hence, it must be the case that (i) $a\not\in\pi(w)$ or (ii) $c\not\in\pi(w)$.
But then since $b\not\in\pi(w)$, in case (i) we have $\pi(u)\cap\pi(w)=\varnothing$ and the votes $u, w$ and candidates $\pi(u)\cup\pi(w)$ form a forbidden $2\times 4$ minor. Similarly,
in case (ii) we have $\pi(v)\cap\pi(w)=\varnothing$ and the votes $v, w$ and candidates $\pi(v)\cup\pi(w)$ form a forbidden $2\times 4$ minor.
\end{proof}
It is interesting to compare the set of forbidden minors 
for the GS domain (Theorem~\ref{thm:gs-minor})
and 
for the {\sc Cat-GS} domain (Theorem~\ref{thm:gs-cat-minor}): 
while the former contains a single $2\times 4$ minor,
the latter contains four $2\times 4$ minors, each of which is obtained
by swapping the two central candidates in $0$, $1$ or $2$
votes of the original minor.

Given this minor-based characterization, we can then apply Theorem~\ref{thm:2axesinP}.
\begin{corollary}\label{cor:gs-cat-easy}
{\sc Cat-GS Voter $2$-Partition} is in {\em P}.
\end{corollary}
However, we cannot use the argument in the proof of 
Theorem~\ref{thm:vr-npcomplete} 
to show that {\sc Cat-GS Voter $3$-Partition}
is hard for $k\ge 3$. This is because a set of
of votes that corresponds to a clique in the input graph may contain a $2\times 4$ forbidden minor for {\sc Cat-GS}.
We provide an explicit example in Appendix~\ref{app:gs-cat}
(Example~\ref{ex:gscat}).
 A similar issue arises if we try to adapt the hardness proof of \citeauthor{ELP13}~[\citeyear{ELP13}] for the SP domain (which was based on $2\times 4$ minors) to the {\sc Cat-GS} domain. 
 Indeed, we cannot rule out the possibility
 that {\sc GS-Cat Voter $k$-Partition} is polynomial-time solvable for 
 $k>2$; however, it does not seem possible to prove this using the proof technique of Theorem~\ref{thm:2axesinP}.
 
 
\section{Conclusion and Future Directions}
Our work contributes to the study of structured and nearly-structured preferences.
We provide a complexity classification for \textsc{GS Voter $k$-Partition}, 
showing that this problem is easy for $k=2$ and hard for each value of $k\ge 3$.
For the domain {\sc Cat-GS}, we describe a simple recognition algorithm
and characterization in terms of forbidden minors that is a natural
consequence of that algorithm. This characterization implies 
that \textsc{Cat-GS Voter $2$-Partition} is in P as well.

The most immediate open problem suggested by our work is the complexity 
of \textsc{Cat-GS Voter $k$-Partition} for $k\ge 3$. We remark that
the complexity of \textsc{SC Voter $k$-Partition} for $k\geq 3$
is open as well \cite{JPE18}. One can also explore other notions
of closeness to group-separability and caterpillar group-separability:
such as, e.g., the minimum number of candidate swaps required to make the input
profile (caterpillar) group-separable; the variants
of this problem for the GS domain where the closeness measure
is based on voter/candidate deletion are NP-complete \cite{BCW16}.

A somewhat different direction is to ask whether a given profile $P$ 
can be split into two subprofiles $P_1$ and $P_2$ so that $P_1$
and $P_2$ belong to two {\em different} domains: e.g., so that
$P_1$ is single-peaked while $P_2$ is group-separable; solving problems of this type may require new proof techniques.


\bibliographystyle{named}
\bibliography{ijcai22}

\newpage

\appendix
\section{Omitted Material from Section~\ref{sec:prelim}}\label{app:prelim}
\begin{proof}[Proof of Proposition~\ref{prop:gs-tree}]
Consider a profile $P$ over a candidate set $C$ that is $T$-consistent for some rooted ordered binary tree $T$ whose leaves are labeled with candidates in $C$;
we will prove that $P$ is group-separable.
For each internal node $x\in T$, let $A(x)$ be the 
labels of the leaves of the subtree rooted at $x$. Let $\ell(x)$ be the left child of $x$ and let $r(x)$ be the right child of $x$. Let $L(x)$ be the set of 
labels of the leaves of the subtree rooted at $\ell(x)$,
and let $R(x)$ be the set of labels of the leaves of the subtree rooted at $r(x)$.

We will now show that for every subset of candidates
$A\subset C$ there is a non-trivial subset $B\subset A$ with $B\neq\varnothing, A$
such that for each $v\in P$ we have $B\succ_v A\setminus B$ or $A\setminus B\succ_v A$. Let $x$ be the maximum-depth node such that $A\subseteq A(x)$.
Then by our choice of $x$ we have 
$A\cap L(x)\neq \varnothing$ and $A\cap R(x)\neq \varnothing$.
Since $P$ is $T$-consistent, for every $v\in P$ we have 
$L(x)\succ_v R(x)$ or $R(x)\succ_v L(x)$, and, consequently, 
$L(x)\cap A\succ_v R(x)\cap A$ or $R(x)\cap A\succ_v L(x)\cap A$.
Thus, it suffices to set $B=A\cap L(x)$, so that $A\setminus B =A\cap R(x)$.

Conversely, suppose $P$ is a profile over a candidate set $C$ that is group-separable. We will construct a tree $T$ recursively as follows. 
If $C$ is a singleton, we create a tree consisting of a single node. 
Otherwise, by definition, 
there exists a set $A\subset C$ with $A\neq\varnothing, C$ such that 
for each voter $v\in P$ we have $A\succ_v C\setminus A$ or $C\setminus A\succ_v A$.
We consider the restriction of $P$ to $A$; this profile
is group-separable, so by the inductive hypothesis it is $T_L$-consistent
for some tree $T_L$. Similarly, the restriction of $P$ to $C\setminus A$
is group-separable and hence $T_R$-consistent for some tree $T_R$.
We now construct a tree $T$ that has root $r$ and two children; the tree rooted
at the left child is $T_L$ and the tree rooted at the right child is $T_R$.
It is immediate that $P$ is $T$-consistent.
\end{proof}

\section{Omitted Material from Section 3}\label{appendixa}

\begin{proof}[Proof of Theorem \ref{thm:2axesinP} (missing part)]

We will now argue that $G$ is bipartite if and only if $P$ can be split into two $X$-consistent subprofiles.

Indeed, suppose that $G$ is bipartite, and let $(U, V)$
be the respective partition of $P$.
Consider the set $U$; the argument for $V$ is similar.
Since $U$ forms an independent set in $G$, it does not contain a $2\times \ell$ forbidden minor. 
Moreover, for each triple $T=\{a, b, c\}$, either $P^i_{T,x}=\varnothing$ for some $x\in T$ or $T$ is $i$-dangerous. This means $P^i_{T,a},P^i_{T,b},P^i_{T,c}\neq \varnothing$ and so $P^i_{T,x}$ is not $X$-consistent for some unique $x\in T$, say $x=c$ and by 
Lemma~\ref{cl:onepertriple} $U\cap P^i_{T,a} = \varnothing$ or $U\cap P^i_{T,b} = \varnothing$.
So there exists an $x\in T$ such that 
$U\cap P^i_{T,x}=\varnothing$, and hence $U$ does not contain a forbidden $i$-minor with candidate set $T$. Hence, $U$ does not contain forbidden
minors of $X$, and is therefore $X$-consistent.

Conversely, suppose that $P$ can be partitioned as $U\cup V$
so that $U$ and $V$ are both $X$-consistent. We claim
that both $U$ and $V$ form independent sets in $G$. 
We prove this claim for $U$; the argument for $V$ is similar. 

If there is a pair of vertices $u, v\in U$
such that $\{u, v\}$ is a $2\times \ell$-edge, then $U$ contains
a $2\times \ell$ forbidden minor, a contradiction.
Now, suppose that $U$ contains an $i$-edge $\{u, v\}$.
This edge by construction corresponds to an $i$-dangerous triple $T=\{a, b, c\}$ with $P^i_{T,c}$  not $X$-consistent; we can assume without loss of generality that $u\in P^i_{T,a}$, 
$v\in P^i_{T,b}$. Since $P^i_{T,c}$ is not $X$-consistent, it cannot be the case that $P^i_{T,c}\subseteq V$ and hence $P^i_{T,c}\cap U\neq\varnothing$; let $w$ be some vote in $P^i_{T,c}\cap U$.
But then $u$, $v$, $w$ and $T$ form a forbidden $i$-minor,
a contradiction.
\end{proof}

\section{Omitted Material from Section~\ref{sec:gs-cat}}\label{app:gs-cat}

\begin{proof}[Proof of Proposition~\ref{prop:gs-cat-charac}]
By definition $v$ is group-separable on $(c_1,\ldots, c_m)$ if for every $j\in [m]$ candidate $c_j$ is ranked either first or last in the restriction of $v$ to $C_{[j, m]}$.
We define $C'$ to be the set of all candidates $c_j\in C$ that are ranked first in the restriction of $v$ to $C_{[j, m]}$. 
Then $C'$ forms an initial segment of $\succ_v$, 
so we obtain $C\succ_v C\setminus C'$.
Further, in $v|_{C'}$ the candidates appear in the increasing order of indices, since if $c_j\in C'$, it is ranked before $\{c_{j+1},\ldots, c_m\}\cap C'$. By the same argument, in $v|_{C\setminus C'}$ the candidates appear in the decreasing order of indices.

Conversely, if for every $v\in V$ there is an initial segment $C'\subset C$ of $\succ_v$ that is ranked by increasing index and a final segment $C\setminus C'$ ranked by decreasing index, then for $c_j\in C'$ we have $c_j\succ_v (C\setminus C')$ and $c_j\succ_v\{c_{j+1},\ldots,c_m\}\cap C'$, so $c_j\succ\{c_{j+1},\ldots,c_m\}$, showing that by definition $v$ is group-separable on the caterpillar $(c_1,\ldots, c_m)$.
\end{proof}

\begin{proof}[Proof of Proposition~\ref{prop:gs-cat-closed}]
Consider a profile $P$ that is group-separable on the caterpillar 
given by $(c_1, \dots, c_m)$. Let $P'$ be the profile obtained
by deleting a candidate $c_j$, $j\in [m]$, from every vote in $P$.
By Proposition~\ref{prop:gs-cat-charac} for every $v\in P$ there exists a $C'$ such that $C'\succ_v C\setminus C'$ and $C'$ is ordered by increasing index and $C\setminus C'$ by decreasing index. But then $v$ restricted to $C\setminus \{c_j\}$ satisfies $C'\setminus\{c_j\}\succ_v C\setminus C'\cup\{c_j\}$, and $C'\setminus\{c_j\}$ is ordered by increasing index, while
$C\setminus C'\cup\{c_j\}$ is ordered by decreasing index. So again, by Proposition~\ref{prop:gs-cat-charac}, $P'$ is group-separable on $(c_1,\ldots c_{j-1},c_j,\ldots c_m)$.

\end{proof}

\begin{proof}[Proof of Theorem 4]
If $P$ is not medium-restricted, then
$P$ is not group-separable by the forbidden minor characterization of group-separable profiles, and so in particular not group-separable on a caterpillar.

Now, suppose $P$ contains votes $u,v$ and candidates $a,b,c,d$ such that
 $a\succ_u\{b,c\}\succ_u d$ and $b\succ_v\{d,a\}\succ_v c$. 
 By Proposition~\ref{prop:gs-cat-closed} the domain {\sc Cat-GS} is closed under candidate 
 deletion, and, by definition, it is closed under voter deletion.
 Hence, if $P$ was group-separable on a caterpillar, then
 so would be the restriction of the profile $(u, v)$ to 
 $\{a,b,c,d\}$. But it is easy to see that the latter profile is not group-separable
 on a caterpillar, since no candidate appears in an extreme (top or bottom) position
 in both votes.
\end{proof}

\begin{example}\label{ex:gscat}
Suppose the graph $G'=(V', E')$ contains a subset of vertices $U=\{u, v, w, x\}$ such that the induced subgraph on $U$ is a matching: $(U\times U) \cap E' = \{\{u, v\}, \{w, x\}\}$.
The reduction in the proof of Theorem~\ref{thm:vr-npcomplete}
will create a voter for each vertex in $U$ and a triple of candidates for each pair of vertices in $(U\times U)\setminus E'$.
Let $\{a, b, c\}$ and $\{d, e, f\}$ be the triples of candidates that correspond to missing edges $\{u, w\}$ 
and $\{v, x\}$, respectively.
Then our reduction constructs a profile where 
$u$ and $v$ disagree on the ranking of $\{a, b, c\}$
(because $u$ is incident to $\{u, w\}$ and $v$ is not), 
and also 
$u$ and $v$ disagree on the ranking of $\{d, e, f\}$
(because $v$ is incident to $\{v, x\}$ and $u$ is not).
Renaming the candidates if necessary, we can assume that
$a\succ_u b$, $b\succ_v a$ and $d\succ_u e$, $e\succ_v d$.
Also, all voters rank $\{a, b, c\}$ above $\{e, d, f\}$.
Hence, $a\succ_u b \succ_u d\succ_u e$, 
$b\succ_v a \succ_v e\succ_v d$ forms a forbidden 
$2\times 4$ minor for {\sc Cat-GS}, even though $\{u, v\}$ is a clique in $G'$. Thus, a clique partition of $G'$ does not necessarily correspond to a partition of the constructed profile $P$ into caterpillar group-separable profiles.
\end{example}

\end{document}